\documentclass[fp]{jpsj3x}
\usepackage{txfonts}
\usepackage{amsmath}
\usepackage{amssymb}
\usepackage{color}
\usepackage{amscd}
\numberwithin{equation}{section}
\newtheorem{Theorem}{Theorem}[section]
\newtheorem{Prop}[Theorem]{Proposition}

\newtheorem{Example}[Theorem]{Example}

\newcommand{\R}{\mathbb{R}}
\newcommand{\Z}{\mathbb{Z}}

\newcommand{\ve}{\varepsilon}
\newcommand{\e}{{\rm e}}

\newcommand\ot{\otimes}
\newcommand{\Pth}{{\mathcal P}}

\newcommand{\wt}{{\rm wt}}

\newcommand{\proof}{\noindent \textit{Proof. }}
\newcommand{\qed}{\hfill $\Box$}

\title{Commuting Time Evolutions in the Tropical Periodic Toda Lattice}

\author{\name{Taichiro \surname{Takagi}}\thanks{E-mail address: takagi@nda.ac.jp}
}
\inst{\address{Department of Applied Physics, National Defense Academy, Yokosuka, Kanagawa 239-8686, Japan} 
}
\abst{The tropical (ultradiscrete) periodic Toda lattice
is a dynamical system 
derived from a time-discretized version of
the periodic Toda lattice
through a limiting procedure called tropicalization.
We propose a new formulation for this dynamical system
with a representation by two-colored strips.
Based on this formulation,
a family of its commuting time evolutions is constructed.}

\kword{solitons, integrable systems, box-ball systems, Yang-Baxter relation, max-plus algebra}

\begin{document}
\maketitle
\section{Introduction}\label{sec:1}
The Toda lattice\cite{Toda} is one of the most important examples of
classical integrable systems that
appears in many topics in modern mathematics and physics.
Recently,
among the growing interest in tropical mathematics\cite{SS04, MZ06, IMS2009},
the {tropical periodic Toda lattice} (trop p-Toda) \cite{IKT11}
(also known as the {ultradiscrete periodic Toda lattice} \cite{KT02})
is attracting attentions \cite{IT07,IT08,IT09,IT09b}.
This is a dynamical system 
derived from a time-discretized version of
the periodic Toda lattice
through a limiting procedure called tropicalization.
The time evolution of this dynamical system is described by the following equations
\begin{align}\label{i:eq:UDpTodax}
\begin{split}
&Q_j^{t+1} = \min(W_j^t, Q_j^t-X_j^t),
\qquad X_j^t = \min_{0 \leq k \leq N-1}
\bigl(\sum_{l=1}^k (W_{j-l}^t - Q_{j-l}^t)\bigr),
\\
&W_j^{t+1} = Q_{j+1}^t+W_j^t - Q_j^{t+1}.
\end{split}
\end{align}
The initial value problem of
this system was solved
in terms of tropical Riemann theta functions. 
However, so far no phase flows of the system have been discovered,
except that for the original time evolution \eqref{i:eq:UDpTodax}.
In this paper 
we construct
a family of commuting flows (generalized time evolutions) in trop p-Toda.
The construction of such flows will be highly significant 
for the future studies of various
tropical/ultradiscrete dynamical systems,
because they can determine the global structure of their level sets
as in the case of Liouville-Arnold's theorem for Hamiltonian systems \cite{Ar}.

A special case of trop p-Toda is known as the periodic box-ball system (pBBS) \cite{YYT},
a one-dimensional cellular automaton.
In this case the theory of Kashiwara's crystals\cite{K} for affine Lie algebra $\widehat{\mathfrak{sl}}_2$ enables us to construct
a family of commuting time evolutions \cite{KTT}, that eventually allows one to
identify its level set structure \cite{T10}.
Indeed, our construction of the phase flows in the
general trop p-Toda also comes from
an important concept in the theory of crystals,
{\em the combinatorial} $R$\cite{KMN1,NY97}.

Let us explain the idea briefly.
Without loss of generality,
we can assume 
all the $Q,W$-variables in
\eqref{i:eq:UDpTodax} take their values in $\mathbb{R}_{>0}$.
This enables us to represent the time evolution of trop p-Toda by
a sequence of two-colored (white and black) strips,
where the lengths of the white (resp.~black) segments are denoted by
$W_j$s (resp.~$Q_j$s).
See Fig.~\ref{f1} in \S \ref{sec:4}.
First we recall the corresponding result in pBBS\cite{KTT}.
Now the values of the $Q,W$-variables in \eqref{i:eq:UDpTodax}
are restricted to $\mathbb{Z}_{>0}$.
Hence the segments in each color can be replaced by 
finite sequences of identical letters, say $1$ for white and $2$ for black. 
For instance, a state of pBBS at time $t$ may be represented as
\begin{equation*}
11\overbrace{222..22}^{Q_1^t}\overbrace{11...1}^{W_1^t}
\overbrace{2...22}^{Q_2^t}\overbrace{1...1}^{W_2^t}
........\overbrace{11...1}^{W_{N-1}^t}\overbrace{22...2}^{Q_N^t}111
\end{equation*}
where the number of $1$'s in the both ends amounts to $W_N^t$.
In this case, the letters $1$ and $2$ are regarded as elements of 
the $\widehat{\mathfrak{sl}}_2$ crystal $B_1$.
For any positive integer $l$, the combinatorial
$R: B_l \otimes B_1 \rightarrow B_1 \otimes B_l$ for
the isomorphism of the $\widehat{\mathfrak{sl}}_2$ crystals
yields a time evolution $T_l$.
The original time evolution \eqref{i:eq:UDpTodax} is given by $T_\infty$, and
the commutativity of the time evolutions $[T_l, T_k] = 0$
follows from the Yang-Baxter relation\cite{KS80} satisfied by the map $R$.

The case of general trop p-Toda is treated as follows.
For any positive real number $l$,
let $\mathbb{B}_l$ be a continuous analogue of the $\widehat{\mathfrak{sl}}_2$ crystal $B_l$.
It is related to the theory of geometric crystals\cite{BK00}, an algebro-geometric analogue of 
the theory of Kashiwara's crystals.
This enables us to interpret
the combinatorial $R$ also as a map 
$R: \mathbb{B}_l \times \mathbb{B}_{l'} \rightarrow \mathbb{B}_{l'} \times \mathbb{B}_l$.
We divide a two-colored strip $S_1$ into black and white segments, and  successively apply
a map $\overline{R}$, a simple modification of $R$,
to the segments with various lengths $l'$. 
Finally we recombine the resulting segments into a new strip $S_2$.
With some ideas to treat the periodic boundary condition,
one gets a family of time evolutions $T_l: S_1 \mapsto S_2$ with the commutativity $[T_l, T_k] = 0$.

This paper is organized as follows.
In \S \ref{sec:2} and \S \ref{sec:3} we give brief reviews on the procedure of tropicalization
and the notions of birational and combinatorial $R$ maps.
In \S \ref{sec:4} we define the trop p-Toda and give its description
in terms of two-colored strips.
In \S \ref{sec:5} we give a brief review on pBBS.
In \S \ref{sec:6} we construct a family of time evolutions in trop p-Toda.
In \S \ref{sec:7} we prove the commutativity of the time evolutions.
The main result of this paper is Theorem \ref{t:july15f}.
Finally, we give several discussions in \S \ref{sec:8}.

\section{Tropicalization}\label{sec:2}
In this section we give an exposition on the procedure of {\em tropicalization}
based on Ref.~\citen{IKT11}.
We restrict ourselves to use this procedure to derive some relevant piecewise-linear equations,
without trying to discuss related notions in tropical geometry such as tropical polynomials or tropical curves\cite{IMS2009}.

Define $\mathbb{T} = \R \cup \{ \infty \}$
where $\infty$ is an element satisfying 
$ a < \infty$ and $\infty + a = \infty$ for any $a \in \R$. 
The algebra $(\mathbb{T}, \oplus, \odot)$ is called the {\it min-plus algebra} 
where the two operations $\oplus$ and $\odot$
are defined by
$$
a \oplus b := \min(a, b), 
\qquad 
a \odot b := a + b.
$$
They are called
{\it tropical addition} and {\it tropical multiplication} respectively.
Obviously one has
$$
a \oplus \infty = a,
\qquad 
a \odot 0 = a,
$$
for any $a \in \mathbb{T}$, 
hence $\infty$ is the additive zero and $0$ is the multiplicative unit. 
The multiplicative inverse of $a(\ne \infty)$ is given by $-a$,
but no $a \in \mathbb{T}$ has its additive inverse. 
This algebra is isomorphic to the {\it max-plus algebra}
$(\mathbb{T}', \oplus', \odot)$ where $\mathbb{T}' = \R \cup \{ -\infty \}$ and $a \oplus' b := \max(a, b)$
\cite{IMS2009}.

We introduce a limiting procedure called the {\it tropicalization}.
Actually, there are two versions of this procedure: the min-plus version and the max-plus version.
The min-plus version of the tropicalization
links the subtraction-free algebra
$(\R_{>0},+,\times)$ to the min-plus algebra $(\mathbb{T}, \oplus, \odot)$. 
We define a map $\mathrm{Log}_\ve : \R_{>0} \to \R$ with an infinitesimal 
parameter $\ve > 0$ by
\begin{align}
  \label{i:loge-map}
  \mathrm{Log}_\ve : a \mapsto - \ve \log a.
\end{align}
For $a > 0$, define $A \in \R$ by $a = \e^{-\frac{A}{\ve}}$.
Then we have $\mathrm{Log}_\ve (a) = A$. 
Moreover, for $a, b > 0$ define $A,B \in \R$ by $a = \e^{-\frac{A}{\ve}}$ and 
$b = \e^{-\frac{B}{\ve}}$. 
Then we have  
$$
  \mathrm{Log}_\ve (a + b) = 
  -\ve \log (\e^{-\frac{A}{\ve}} + \e^{-\frac{B}{\ve}}),
  \quad
  \mathrm{Log}_\ve (a \times b) = A + B.
$$
In the limit $\ve \to 0$, $\mathrm{Log}_\ve (a + b)$ becomes $\min(A,B)$.
The procedure $\lim_{\ve \to 0} \mathrm{Log}_\ve$ 
with the transformation as $a = \e^{-\frac{A}{\ve}}$
is called (the min-plus version of) the tropicalization.
%
Through this procedure,
subtraction-free rational equations on $\R_{>0}$ reduce to 
piecewise-linear equations on $\R$ described by min-plus algebra.
See Ref.~\citen{IKT11} for more details, and a definition of the max-plus version
of the tropicalization.

%

We remark that these procedures are conventionally called {\em ultradiscretizations}.
Since the variables $A, B$ in the above are not required to take discrete values,
we reserve the word ultradiscretization for
a special case of the tropicalizations\cite{IKT11}. 

\section{Birational $R$ and Combinatorial $R$}\label{sec:3}
The {\em birational} $R$ \cite{IKT11,Y01} is a notion related to the theory of geometric crystals\cite{BK00}. 
In this paper
the birational $R$ for affine Lie algebra $\widehat{\mathfrak{sl}}_N$ has two roles.
While it gives an explicit formula for the solution of a discrete periodic Toda lattice equation for general $N$,
its tropicalization (with a slight modification) for $N=2$ will be used to define a local time evolution rules
in trop p-Toda.

Let ${\mathcal B} = \{x = (x_1,\ldots, x_{N})  \} \subset ({\mathbb R}_{>0} )^{N}$ be a set of variables.
The birational $R$
for $\widehat{\mathfrak{sl}}_{N}$ is the birational map 
${\mathcal R}: {\mathcal B} \times {\mathcal B} \rightarrow {\mathcal B} \times {\mathcal B}$ specified by
${\mathcal R}(x,y) = (\tilde{y},\tilde{x})$ in which
\begin{equation}\label{t:eq:RPP}
\begin{split}
&\tilde{x}_i=x_i \frac{P_{i-1}(x,y)}{P_i(x,y)}, \quad \tilde{y}_i=y_i \frac{P_i(x,y)}{P_{i-1}(x,y)},\\
&P_i(x,y)=\sum_{k=1}^{N} \left( \prod_{j=k}^{N} x_{i+j} \prod_{j=1}^k y_{i+j} \right),
\end{split}
\end{equation}
where all the indices are considered to be in $\Z_{N}$.
It satisfies the inversion relation ${\mathcal R}^2=id$ on ${\mathcal B} \times {\mathcal B}$ and the 
Yang-Baxter relation 
\begin{equation}\label{t:YBtropR}
{\mathcal R}_1{\mathcal R}_2{\mathcal R}_1 = {\mathcal R}_2{\mathcal R}_1{\mathcal R}_2
\end{equation}
on ${\mathcal B} \times {\mathcal B} \times {\mathcal B}$,
where ${\mathcal R}_1(x,y,z) = ({\mathcal R}(x,y),z)$ and 
${\mathcal R}_2(x,y,z) = (x,{\mathcal R}(y,z))$. 
The birational $R$ 
is characterized as the unique solution to the following equations
\begin{equation}\label{t:eq:toda}
x_iy_i = \tilde{y}_i \tilde{x}_i,\qquad 
\frac{1}{x_i}+\frac{1}{y_{i+1}} = 
\frac{1}{\tilde{y}_i}+\frac{1}{\tilde{x}_{i+1}},
\end{equation}
with an extra constraint
$\prod_{i=1}^{N}(x_i/\tilde{x}_i) = 
\prod_{i=1}^{N}(y_i/\tilde{y}_i)= 1$.

Consider the max-plus version of the tropicalization of
the birational $R$.
For $l > 0$ let $\mathbb{B}_l = \{ (X_1, \ldots, X_N) \in (\R_{\geq 0})^N | \sum_{i=1}^N X_i = l \}$.
Given $l,l' > 0$, $X = (X_1, \ldots , X_{N} ) \in \mathbb{B}_l$ and $Y = (Y_1, \ldots , Y_{N} ) \in \mathbb{B}_{l'}$, let 
$\tilde{X} = (\tilde{X}_1, \ldots , \tilde{X}_{N} )$,
$\tilde{Y} = (\tilde{Y}_1, \ldots , \tilde{Y}_{N} )$
be defined by 
\begin{equation}\label{t:udRPP}
\begin{split}
&\tilde{X}_i = X_i - P_i(X,Y) + P_{i-1}(X,Y), \quad
\tilde{Y}_i = Y_i + P_i(X,Y) - P_{i-1}(X,Y), \\
&P_i(X,Y) = \max_{1 \leq k \leq N} \left( \sum_{j=k}^{N} X_{i+j} + \sum_{j=1}^{k} Y_{i+j} \right).
\end{split}
\end{equation}
Then we have:
\begin{Prop}
The $\tilde{X}_i, \tilde{Y}_i$ are non-negative.
\end{Prop}
\begin{proof}
Let $Z_{i,k} = \sum_{j=k}^{N} X_{i+j} + \sum_{j=1}^{k} Y_{i+j}$.
To prove $\tilde{X}_i \geq 0$, it suffices to show that for any $1 \leq k \leq N$ there exists $1 \leq k' \leq N$ such that
$X_i -Z_{i,k} + Z_{i-1,k'} \geq 0$.
If one takes $k' = N$ for $k=N$ and $k'=k+1$ for $1 \leq k < N$, then this condition is satisfied.
To prove $\tilde{Y}_i \geq 0$, it suffices to show that for any $1 \leq k \leq N$ there exists $1 \leq k' \leq N$ such that
$Y_i +Z_{i,k'} - Z_{i-1,k} \geq 0$.
Here $k' = 1$ for $k=1$ and $k'=k-1$ for $1 < k \leq N$ satisfy the condition.
\qed
\end{proof}
\par\noindent
Thus one can define a map $R: \mathbb{B}_l \times \mathbb{B}_{l'} \rightarrow \mathbb{B}_{l'} \times \mathbb{B}_l$ by
$R(X,Y) = (\tilde{Y},\tilde{X})$.
This $R$ satisfies the inversion relation and the Yang-Baxter relation on
$\mathbb{B}_l \times \mathbb{B}_{l'}$ and $\mathbb{B}_l \times \mathbb{B}_{l'} \times \mathbb{B}_{l''}$ respectively.
We call this map a {\em combinatorial} $R$ as in the theory of crystals.

For applications to trop p-Toda,
we give an interpretation of the elements of the set $\mathbb{B}_l$.
Suppose there are $N$ kinds of colors.
Regard $X = (X_1, \ldots , X_{N} ) \in \mathbb{B}_l$ as a (at most) $N$-colored strip of length $l$
made by concatenating strips of length $X_i$ with color $i$.
In this paper we restrict ourselves to the $N=2$ case. 

\section{Discrete Periodic Toda Lattice and Its Tropicalization}\label{sec:4}
The discrete periodic Toda lattice (dp-Toda) is given by the evolution equations
\begin{align}\label{i:eq:Toda-q}
\begin{split}
&q_j^{t+1} = q_j^t + w_j^t - w_{j-1}^{t+1} ,
\\
&w_j^{t+1} = \frac{q_{j+1}^t w_j^t}{q_{j}^{t+1}},
\end{split}
\end{align}
with the periodic boundary condition $q_{j+N}^t = q_j^t$, $w_{j+N}^t = w_j^t$.
This is a dynamical system with a discretized time variable and from which one can derive the original Toda lattice\cite{Toda}
as a continuum limit \cite{IKT11}.
Note that, \eqref{t:eq:toda} 
is equivalent to \eqref{i:eq:Toda-q}
under the substitution of variables
$q^t_{j+1}=1/x_j, w^t_j = 1/y_j, q^{t+1}_j = 1/{\tilde x}_j, 
w^{t+1}_j = 1/{\tilde y}_j$.

The tropical periodic Toda lattice (trop p-Toda) is given by 
the piecewise-linear evolution equations
\begin{align}\label{i:eq:UDpToda}
\begin{split}
&Q_j^{t+1} = \min(W_j^t, Q_j^t-X_j^t),
\qquad X_j^t = \min_{0 \leq k \leq N-1}
\bigl(\sum_{l=1}^k (W_{j-l}^t - Q_{j-l}^t)\bigr),
\\
&W_j^{t+1} = Q_{j+1}^t+W_j^t - Q_j^{t+1},
\end{split}
\end{align}
on the phase space
$\displaystyle{
\mathcal{T} 
= 
\{(Q_j,W_j)_{j \in \Z / N\Z} ~|~ \sum_{j=1}^N Q_j < \sum_{j=1}^N W_j \}
\subset \R^{2N}.}
$

\begin{Prop}
Trop p-Toda \eqref{i:eq:UDpToda} is obtained 
as a tropicalization of dp-Toda
\eqref{i:eq:Toda-q} with
the condition
$\prod_{j=1}^N w_j^t / q_j^t < 1$ so that the tropicalization of 
$(1- \prod_{j=1}^N w_j^t / q_j^t)$ is zero.
\end{Prop}
This proposition is due to Ref.~\citen{KT02}.
For a proof through the birational $R$ \eqref{t:eq:RPP}, see 
Proposition 6.13 of Ref.~\citen{IKT11}.

Let $\tilde{Q}_j^t = Q_j^t + C, \tilde{W}_j^t = W_j^t + C$ where $C$ is a constant.
Then the new variables $\tilde{Q}_j^t,\tilde{W}_j^t$ also satisfy the equation 
\eqref{i:eq:UDpToda}.
In addition, we have:
\begin{Prop}
If $Q_j^t,W_j^t >0$ for some $t$ and for all $j$, then trop p-Toda
\eqref{i:eq:UDpToda} ensures $Q_j^{t+1},W_j^{t+1} >0$ for all $j$.
\end{Prop}
\begin{proof}
Since
\begin{equation*}
-X_j^t =  \max_{0 \leq k \leq N-1}
\bigl(\sum_{l=1}^k (Q_{j-l}^t - W_{j-l}^t)\bigr) = 
\max \left[ \{ 0 \} \cup \max_{1 \leq k \leq N-1}
\bigl(\sum_{l=1}^k (Q_{j-l}^t - W_{j-l}^t)\bigr) \right] \geq 0,
\end{equation*}
one can deduce $Q_j^{t+1} > 0$ from the first equation in \eqref{i:eq:UDpToda}.
Also, since
\begin{equation*}
W_j^t - Q_j^{t+1} = W_j^t + \max( -W_j^t, X_j^t-Q_j^t)
= \max (0, W_j^t + X_j^t-Q_j^t) \geq 0,
\end{equation*}
we have $W_j^{t+1} > 0$ through its last equation.
\qed
\end{proof}

Therefore, without loss of generality,
we can restrict its phase space as
$\displaystyle{
\mathcal{T} \subset (\R_{>0})^{2N}.}$
%
This allows us to represent the time evolution of trop p-Toda by
a sequence of two-colored strips:
they are concatenations of, say white and black, segments
in which the lengths of the white (resp.~black) segments are equal to
$W_j$s (resp.~$Q_j$s).
See Fig.~\ref{f1} for an example.
Note that $\sum_{j=1}^N Q_j^t$ and $\sum_{j=1}^N W_j^t$ are conserved
quantities.

\begin{figure}[hbtp]
\centering
\scalebox{1}[1]{
\includegraphics[clip]{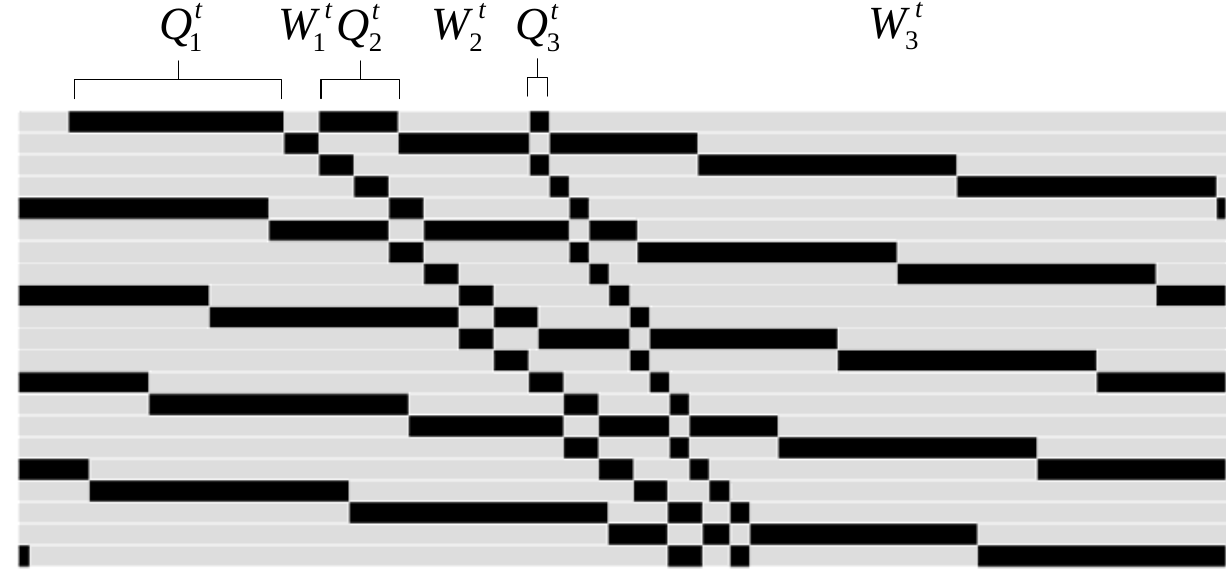}}
\caption{A representation for the time evolution of trop p-Toda \eqref{i:eq:UDpToda}.
Let $t$ be the time for the top row, where 
the lengths of three black segments (``solitons'') are $Q_1^t, Q_2^t, Q_3^t$,
and the lengths of white segments between them 
(with regard to the periodic boundary) are $W_1^t, W_2^t, W_3^t$.
In the second row
the length of the first black segment to the right of
the $j$-th soliton at time $t$ is  $Q_j^{t+1}$. 
In the same way, the strip at the $n$-th row visualizes the quantities   $Q_j^{t+n-1}$ and $W_j^{t+n-1}$.}
\label{f1}
\end{figure}

In fact, the evolution equations \eqref{i:eq:UDpToda} can only determine relative positions of the segments on the strips but
contain no data of their absolute positions.
In this paper we fix those positions in the same way as in the case of pBBS.
In what follows we shall give a formulation of trop p-Toda with this representation by
two-colored strips, where the condition $\sum_{j=1}^N Q_j^t < \sum_{j=1}^N W_j^t$ on the phase space is not always imposed.
\section{Periodic Box-Ball System}\label{sec:5}
As we have explained in \S \ref{sec:1}, the periodic box-ball system (pBBS) is a one-dimensional cellular automaton given by the 
evolution equations \eqref{i:eq:UDpToda} with all the $Q,W$ variables taking their values in $\Z_{>0}$.
In Ref. \citen{KTT} a formulation for pBBS based on the theory of crystals was given.
In order to illustrate background ideas for
our new formulation for trop p-Toda, 
we give a brief review on it.

Let $B_l$ be the crystal of the $l$-fold symmetric 
tensor representation of the quantized affine Lie algebra
$U_q(\widehat{\mathfrak{sl}}_2)$.
As a set it is given by 
$B_l = \{X=(X_1,X_2) \in (\Z_{\ge 0})^2 \mid X_1+X_2=l\}$.
The element $(X_1,X_2)$ will also be expressed as 
the length $l$ row shape semistandard tableau 
where the letter $i$ appears $X_i$ times.
For example,
$B_1=\big\{\fbox{1},\fbox{2}\big\}, 
B_2 = \big\{\fbox{11}, \fbox{12}, \fbox{22} \big\}$.
For two crystals $B$ and $B'$, one can define their tensor product
$B\ot B'=\{b\ot b'\mid b\in B,b'\in B'\}$. 
More general tensor products $B_{l_1}\otimes \cdots \otimes B_{l_k}$
can be defined by using it repeatedly.
The unique isomorphic map of the crystals 
$B_l \otimes B_k \overset{\sim}{\rightarrow}
B_k \otimes B_l$ 
is determined by intertwining properties associated with the actions of
raising and lowering operators.
It is an analogue of the quantum $R$ matrix for $U_q(\widehat{\mathfrak{sl}}_2)$, and
called the combinatorial $R$.
Explicitly it is given by
$R:  X\otimes  Y\mapsto
\tilde{Y} \ot \tilde{X}$ with 
\begin{align}\label{eq:h}
\begin{split}
&{\tilde X}_i = X_i-P_i(X,Y)+P_{i-1}(X,Y),\quad 
{\tilde Y}_i = Y_i-P_{i-1}(X,Y)+P_i(X,Y),\\
&P_i(X,Y) = l+k-\min (X_{i+1}, Y_i),
\end{split}
%
\end{align}
where all the indices are in $\Z_2$.
This expression for the combinatorial $R$ is equivalent to the $N=2$ case of (\ref{t:udRPP}).
The relation is depicted as

\begin{equation*}
\begin{picture}(90,40)(10,-9)
\put(0,10){\line(1,0){20}}\put(-7,7){$X$}\put(22,7){${\tilde X}$}
\put(10,0){\line(0,1){20}}\put(7,24){$Y$}\put(7,-12){${\tilde Y}$}
\put(47,8){or}
\put(80,10){\line(1,0){20}}\put(73,7){${\tilde Y}$}\put(102,7){$Y$.}
\put(90,0){\line(0,1){20}}\put(87,24){${\tilde X}$}\put(87,-11){$X$}
\end{picture}
\end{equation*}
For example $B_l \otimes B_1 \simeq B_1 \otimes B_l$ 
is listed as in Fig.~\ref{fig:cr}.
%
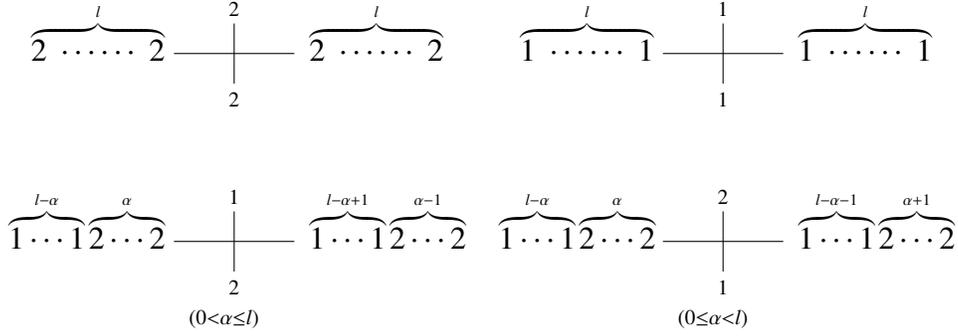
\begin{figure}[htb]
\begin{center}
\begin{picture}(130,120)(55,0)
\unitlength 1mm
\multiput(0,0)(0,25){2}{
\multiput(2,10)(65,0){2}{\line(1,0){16}}
\multiput(10,6)(65,0){2}{\line(0,1){8}}}

\put(9.3,40){$\scriptstyle 2$}
\put(-17,34){$\scriptstyle \overbrace{2 \;\cdots \cdots \; 2}^l$} 
\put(20,34){$\scriptstyle \overbrace{2 \;\cdots \cdots \; 2}^l$}
\put(9.3,28){$\scriptstyle 2$}

\put(74.3,40){$\scriptstyle 1$}
\put(48,34){$\scriptstyle \overbrace{1 \;\cdots \cdots \; 1}^l$} 
\put(85,34){$\scriptstyle \overbrace{1 \;\cdots \cdots \; 1}^l$}
\put(74.3,28){$\scriptstyle 1$}

\put(9.3,15){$\scriptstyle 1$}
\put(-20,9){$\scriptstyle \overbrace{1 \cdots 1}^{l-\alpha}
\overbrace{2\cdots 2}^\alpha$} 
\put(20,9){$\scriptstyle \overbrace{1 \cdots 1}^{l-\alpha+1}
\overbrace{2\cdots 2}^{\alpha-1}$}
\put(9.3,3){$\scriptstyle 2$}
\put(4,-1){$\scriptstyle (0<\alpha\le l)$}

\put(74.3,15){$\scriptstyle 2$}
\put(45.2,9){$\scriptstyle \overbrace{1 \cdots 1}^{l-\alpha}
\overbrace{2\cdots 2}^\alpha$} 
\put(85,9){$\scriptstyle \overbrace{1 \cdots 1}^{l-\alpha-1}
\overbrace{2\cdots 2}^{\alpha+1}$}
\put(74.3,3){$\scriptstyle 1$}
\put(69,-1){$\scriptstyle (0\le \alpha<l)$}

\end{picture}
\caption{Combinatorial $R: B_l \otimes B_1 \simeq B_1 \otimes B_l$}
\label{fig:cr}
\end{center}
\end{figure}
%
The combinatorial $R$ satisfies the Yang-Baxter relation:
\begin{equation}\label{eq:ybr}
(1 \otimes R)(R \otimes 1)(1 \otimes R) =
(R \otimes 1)(1 \otimes R)(R \otimes 1) 
\end{equation}
on $B_j \otimes B_l \otimes B_k$.

We fix the integer $L \in \Z_{\ge 1}$ 
corresponding to the system size.
Set 
\begin{equation}\label{eq:pdef}
\Pth = B_1^{\otimes L}.
\end{equation}
The elements of $\Pth$ are called paths.
%
The periodic box-ball system is 
a dynamical system on $\Pth$ with a family of
time evolutions $T_1, T_2, \ldots$ whose existence is assured by\cite{KTT}:

\begin{Prop}\label{pr:tl}
For any path $p=b_1\otimes \cdots \otimes b_L \in \Pth$ 
and $l \in \Z_{\ge 1}$, there exists an element 
$v_l \in B_l$ such that 
\begin{equation}\label{eq:bpe}
v_l \otimes (b_1 \otimes \cdots \otimes b_L) 
\simeq 
(b'_1 \otimes \cdots \otimes b'_L) \otimes 
v_l
\end{equation}
for some $p' = b'_1 \otimes \cdots \otimes b'_L
\in \Pth$ under the isomorphism 
$B_l\otimes \Pth \simeq 
\Pth \otimes B_l$. 
Such $p'$ is uniquely determined
even if the possible choice of $v_l$ is not unique.
\end{Prop}

For any $l \in \Z_{\ge 1}$ we define $T_l: \Pth \rightarrow \Pth$ by
$T_l (b_1\otimes \cdots \otimes b_L) = b'_1 \otimes \cdots \otimes b'_L$, where the right hand side
is given by \eqref{eq:bpe}.

\begin{Example}\label{t:july15c}
The time evolutions of $p=222111211111$ by $T_l$ with 
$l \geq 3, T_2$ and $T_1$:
\par\noindent
\begin{verbatim}
      t=0  222...2.....  |  222...2.....  |  222...2.....
      t=1  ...222.2....  |  ..222..2....  |  .222...2....
      t=2  ......2.222.  |  ....222.2...  |  ..222...2...
      t=3  22.....2...2  |  ......22.22.  |  ...222...2..
      t=4  ..222...2...  |  2.......2.22  |  ....222...2.
      t=5  .....222.2..  |  222......2..  |  .....222...2
      t=6  2.......2.22  |  ..222.....2.  |  2.....222...
\end{verbatim}

Here we denoted by a dot ``$.$'' for the letter $1$.
\end{Example}

By using the Yang-Baxter relation \eqref{eq:ybr}, one can prove that
these time evolutions are commutative $[T_k, T_l] = 0$ \cite{KTT}.
The time evolution of the original pBBS \cite{YYT} is given by $T_\infty$.
 
The name of a {\em box-ball system} comes from an interpretation of this dynamical system
which we now present.
There is a carrier of balls which travels along an array of boxes, where each box accommodates one ball.
More precisely, one regards $1 \in B_1$ as an empty box and $2 \in B_1$ a ball within a box.
Every element of $B_l$ in Fig.~\ref{fig:cr} is regarded as a carrier of capacity $l$:
if $X=(X_1,X_2) \in B_l$, then it has $X_2$ balls in it.
The carrier travels from the left to the right while picking up/putting a ball out of/into the box if possible. 
Our new formulation of trop p-Toda can be regarded as a continuous analogue of this picture.
\section{Construction of the Time Evolutions}\label{sec:6}
We have introduced a description of trop p-Toda by means of two-colored strips
at the end of \S \ref{sec:4}. 
In this section we define a family of time evolutions in trop p-Toda 
by using this description.
This gives a generalization of the formulation for pBBS in Ref.~\citen{KTT}.

Let $X, Y$ be a pair of strips of length $l$ and $w$.
We assume that the $Y$ is either white or black, while
the $X$ has generally two parts, white in the left and black in the right.
When the length of the black part of $X$ is $x$,
we denote it by $X = (l-x,x)$.

Let 
$\overline{R}:  (X,   Y) \mapsto ( \tilde{Y} , \tilde{X})$ 
be a transformation between pairs of strips 
defined as in Fig.~\ref{f2}, where $\tilde{X}, \tilde{Y}$ are strips of length $l$ and $w$ respectively.
Here we arranged $X,Y,\tilde{X},\tilde{Y}$ around each vertex as
\begin{equation*}
\begin{picture}(90,40)(10,-9)
\put(0,10){\line(1,0){20}}\put(-7,7){$X$}\put(22,7){${\tilde X}.$}
\put(10,0){\line(0,1){20}}\put(7,24){$Y$}\put(7,-12){${\tilde Y}$}
\end{picture}
\end{equation*}
\begin{figure}[hbtp]
\centering
\scalebox{0.6}[0.6]{
\includegraphics[clip]{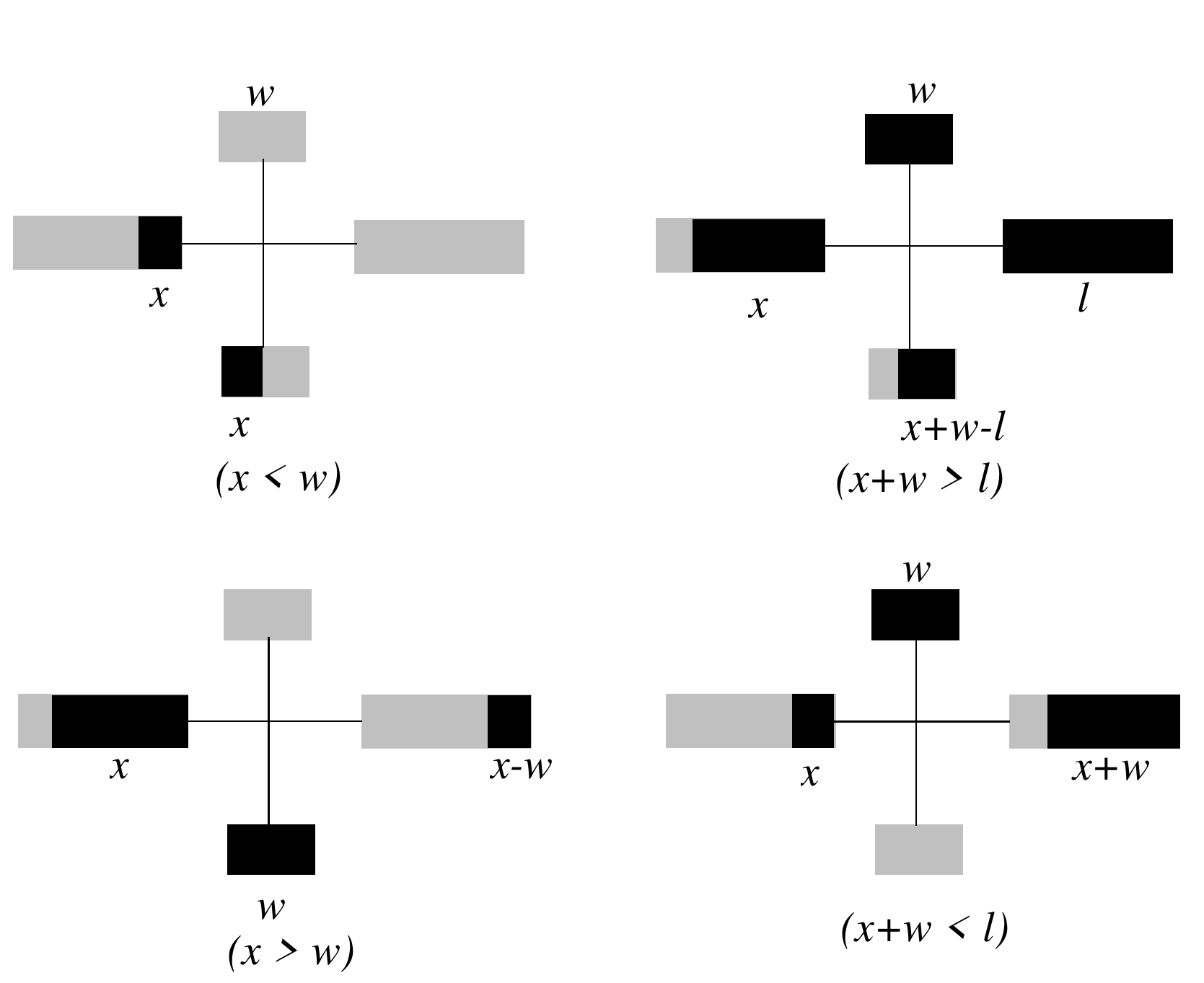}}
\caption{Definition of the transformation $\overline{R}:  (X,   Y) \mapsto ( \tilde{Y} , \tilde{X})$.
Except the top left one where the left part of 
the bottom strip is black, the actions of this $\overline{R}$ can be interpreted as that of the $\widehat{\mathfrak{sl}}_2$ case of the 
combinatorial $R: \mathbb{B}_l \times \mathbb{B}_w \rightarrow \mathbb{B}_w \times \mathbb{B}_l$ in \S \ref{sec:3}.}
\label{f2}
\end{figure}
By using this transformation, we define a time evolution $T_l$ 
for any $l > 0$ in trop p-Toda.

First we explain the procedure to define $T_l$
along an example.
See Figs.~\ref{f2x} and \ref{f2xx}. 
We divide a two-colored strip $S_1$ into black and white segments, and apply
the transformation $\overline{R}$ to each segment repeatedly.
Finally we concatenate the resulting segments into a new strip $S_2$.
Define $T_l$ by $T_l: S_1 \mapsto S_2$.
The time evolution of the original trop p-Toda in Fig.~\ref{f1} is given by $T_{\infty}$. 
By an analogy with pBBS, we call the strips on the horizontal line in Fig.~\ref{f2x} the
{\em carriers} of capacity $l$.
An important point here is that one can always make the contents of the two carriers at the both ends
identical as required to satisfy the periodic boundary condition.
As in the case of pBBS, this time evolution is invertible.
Hence the backward time evolution $T_l^{-1}$ is also uniquely determined.

\begin{figure}[hbtp]
\centering
\scalebox{1}[1]{
\includegraphics[clip]{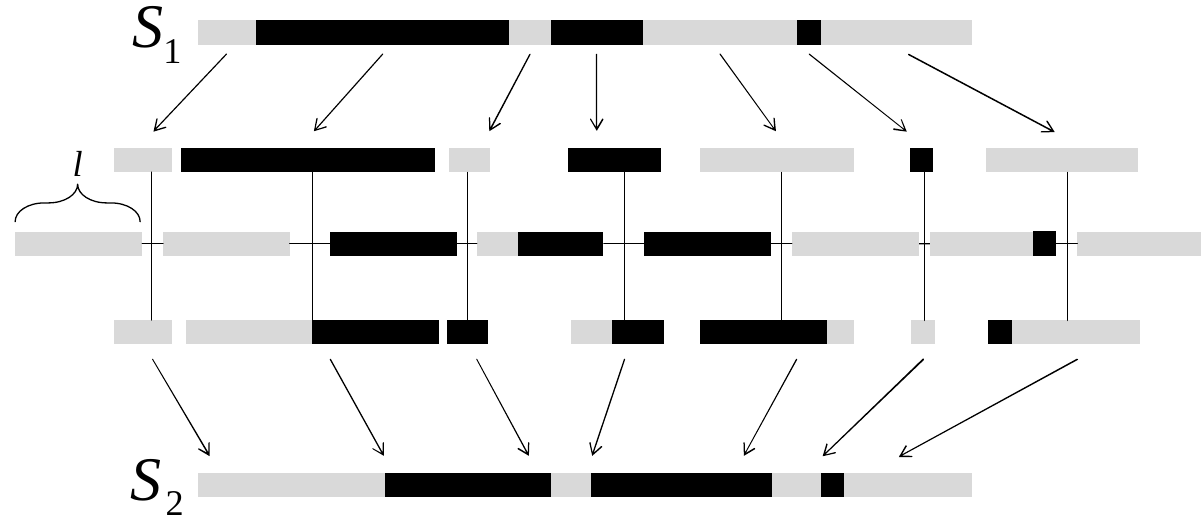}}
\caption{Definition of the time evolution $T_l:S_1 \mapsto S_2$ by division of a strip into segments.}
\label{f2x}
\end{figure}

\begin{figure}[hbtp]
\centering
\scalebox{1}[1]{
\includegraphics[clip]{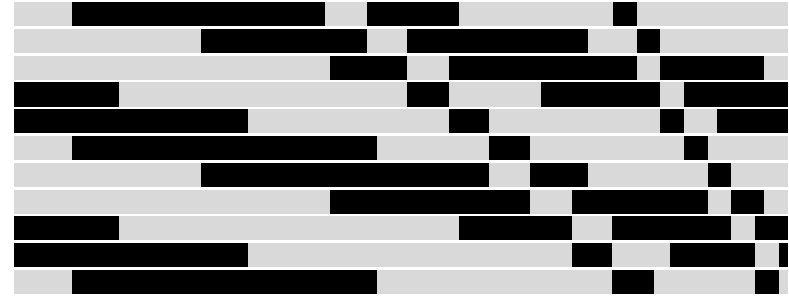}}
\caption{Successive time evolutions of trop p-Toda by a carrier of capacity $l$.
The transition from the first line to the second line is given by $T_l: S_1 \mapsto S_2$ in Fig.~\ref{f2x}.}
\label{f2xx}
\end{figure}

Let us explain the definition of $T_l:S_1 \mapsto S_2$ more in detail.
Denote by $b_1, \ldots, b_K$ the black and the white segments obtained from the strip $S_1$ after the division.
Denote by $X = (X_1, X_2)$ the carrier of capacity $l$ at the left end in Fig.~\ref{fig:xp2}.
As is depicted there,
we apply the transformation $\overline{R}$ repeatedly and obtain
a sequence of segments $b'_1, \ldots, b'_K$ and a carrier $Y = (Y_1,Y_2)$ at the right end.
Let $S_2$ be the strip obtained by concatenating the segments $b'_1, \ldots, b'_K$.
The following Theorem \ref{th:may2_1} allows one to make $X=Y$ and to determine the associated strip $S_2$ uniquely.

We denote by $L$ the common length of the strips $S_1$ and $S_2$.
Let $w (S_i) \, (i=1,2)$ be the sum of the lengths of the white segments in $S_i$, and 
$b (S_i) = L - w (S_i)$ that for the black segments.
Let $\wt (S_i) = w (S_i) - b (S_i)$ be the weight of $S_i$.
Now we present:
\begin{Theorem}\label{th:may2_1}
Let $l$ be any positive real number, and
$S_1$ be any two-colored (black and white) strip.
Given $X = (X_1,X_2)$,
let $S_2,Y = (Y_1, Y_2)$ be those defined in the above procedure.
Then there exists a real number $0 \leq \eta \leq l$ such that
if $X_2 = \eta$, then $Y_2 = \eta$.
Such $\eta$ is unique except $\wt(S_1)=0$ case, where 
$S_2$ 
is independent of the possibly non-unique choice of $\eta$.
\end{Theorem}
\begin{proof}
The top two types of the vertices in Fig.~\ref{f2}
can be split as in Fig.~\ref{f3}.
Given $S_1$ and $X$, we split each vertex in Fig.~\ref{fig:xp2} accordingly.
By renewing notations, we denote the resulting sequences by
$b_1, \ldots, b_K$ and $b'_1, \ldots, b'_K$ again.
Obviously, this procedure does not affect $S_2$ and $Y$.
Note that after this splitting procedure, every vertex in Fig.~\ref{fig:xp2}
is one of those depicted in Fig.~\ref{f4}.
Hence they can be interpreted as the $\widehat{\mathfrak{sl}}_2$ case of the combinatorial $R$ in \S \ref{sec:3}.

Regard 
$Y_2=Y_2(\xi)$ as a function of $\xi = X_2$.
Then $Y_2: [0,l] \rightarrow [0,l]$ is a non-decreasing continuous piecewise-linear function, and
the slope of its graph is either $0$ or $1$.
Hence the intersection of its graph and the line $Y_2 = \xi$ is
either a point or a line segment.

Set $c=Y_2(0)$.
If $Y_2(l) = c$, then it is clear that $\xi = c$ is the unique solution of $Y_2 = X_2$.
Suppose otherwise.
Then there exists $0 \leq a < l$ such that $Y_2(a)=c$ and
$Y_2(\xi)$ is strongly increasing at $\xi = a$.
This condition is satisfied only if 
every vertex in Fig.~\ref{fig:xp2} (after the splitting of the vertices as described above)
with $(X_1,X_2) =(l-a,a)$ falls into the following two types: 
either the bottom left type with $x < l$, or
the bottom right type with $x < l-w$ in Fig.~\ref{f4}.
This implies 
$\wt(S_2)=-\wt(S_1)$.
There are three cases.

{\em Case 1:}
Let $M = b (S_1)$.
We first consider the case $M<L-M$, i.~e.~$\wt(S_1)>0$.
Since $\wt(S_2)=-\wt(S_1)$, one has $b (S_2) = L-M$ for $(X_1,X_2)=(l-a,a)$ in Fig.~\ref{fig:xp2}.
Then the conservation of the total lengths of the black segments demands that
$a+M=c+L-M$. By the assumption $M<L-M$,
we have $a > c$.
Then by the reason on the intersection of the graphs
one easily sees that the unique solution of $Y_2 = X_2$ is $\xi = c$.

{\em Case 2:}
The case $\wt (S_1) < 0$ can be shown similarly by
interchanging the role of black and white.

{\em Case 3:}
Next we consider the case $\wt(S_1)=0$.
Define $a$ and $c$ as above.
Then the condition $\wt(S_2)=-\wt(S_1)$ demands $a = c$.
In this case the intersection of the graphs is a line segment, hence 
there exists $c < c' \leq l$ such that
we have $Y_2(\eta) = \eta$ for $ \eta \in [c,c']$.
This implies that
every vertex in Fig.~\ref{fig:xp2} with $X_2 = \eta \in [c,c']$ is one of
the bottom two types in Fig.~\ref{f4}. 
Namely, for any $X_2 \in [c,c']$ 
every pair $(b_i, b'_i)$ in Fig.~\ref{fig:xp2} has opposite colors.
This verifies the existence of the solution of $Y_2 = X_2$ and the uniqueness of the associated strip $S_2$.
\qed
\end{proof}

\begin{figure}[h]
\begin{picture}(50,50)(-140,5)
\unitlength 1mm
\put(0,10){\line(1,0){20}}\put(33,10){\line(1,0){10}}
\put(24.5,9){$\cdots$}
\put(-22,9){$\scriptstyle \overbrace{\phantom{1 \cdots 1}}^{X_1}
\overbrace{\phantom{2\cdots 2}}^{X_2}$}
\put(-31,9){$X=$}
\put(68,9){$=Y$}
\put(-22,9){\framebox(10.5,3){}}
\put(-11.5,9){\framebox(10.5,3){}}
\multiput(-11.5,9)(0.5,0){21}{\line(0,1){3}}
\multiput(-11.5,9)(0,0.5){6}{\line(1,0){10.5}}
\put(45,9){$\scriptstyle \overbrace{\phantom{1 \cdots 1}}^{Y_1}
\overbrace{\phantom{2\cdots 2}}^{Y_2}$}
\put(45,9){\framebox(10.5,3){}}
\put(55.5,9){\framebox(10.5,3){}}
\multiput(55.5,9)(0.5,0){21}{\line(0,1){3}}
\multiput(55.5,9)(0,0.5){6}{\line(1,0){10.5}}
\put(5,6){\line(0,1){8}} 
\put(4,16){$\scriptstyle b_1$}\put(4,2.2){$\scriptstyle b'_1$}
\put(13,6){\line(0,1){8}}
\put(12,16){$\scriptstyle b_2$}\put(12,2.2){$\scriptstyle b'_2$}
\put(38,6){\line(0,1){8}}
\put(37,16){$\scriptstyle b_K$}\put(37,2.2){$\scriptstyle b'_K$}
\end{picture}
\caption{Diagram to define time evolutions in trop p-Toda.}
\label{fig:xp2}
\end{figure}
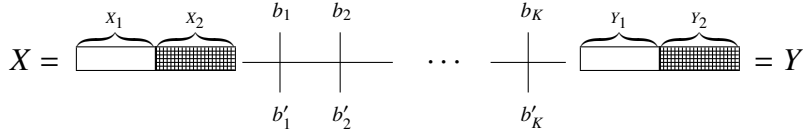

\begin{figure}[hbtp]
\begin{center}
\scalebox{1}[1]{
\includegraphics[clip]{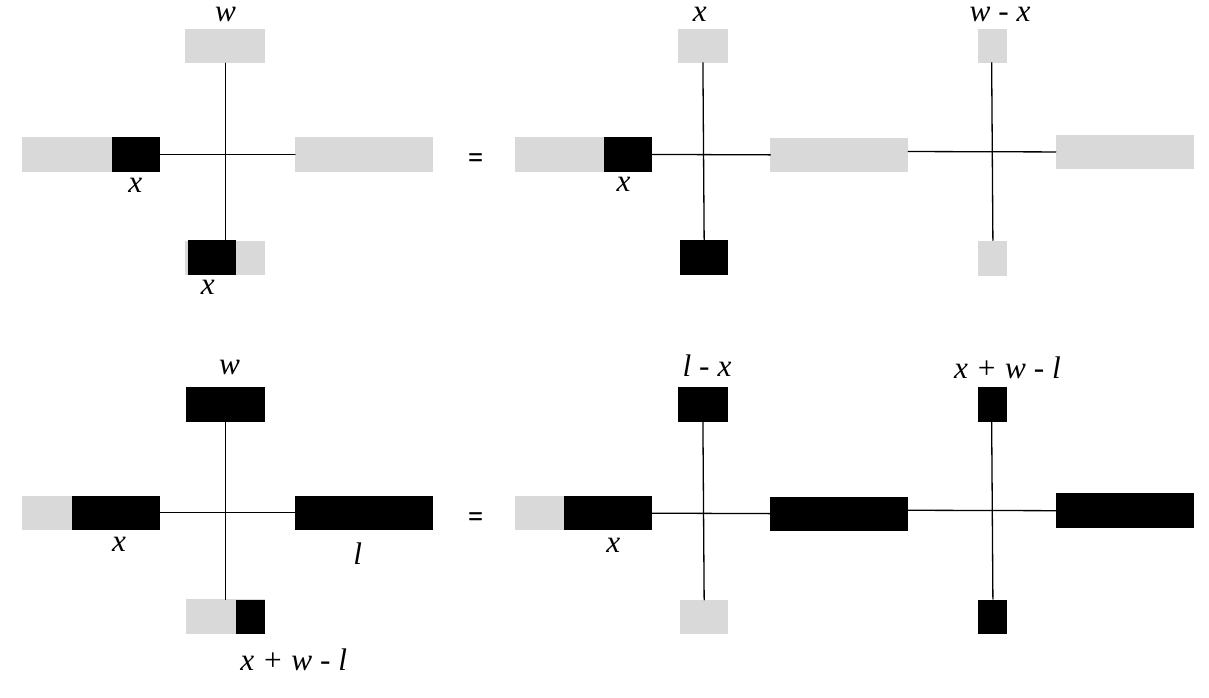}}
\end{center}
\caption{Fine splitting of the local time evolution rules.}
\label{f3}
\end{figure}

\begin{figure}[hbtp]
\begin{center}
\scalebox{0.6}[0.6]{
\includegraphics[clip]{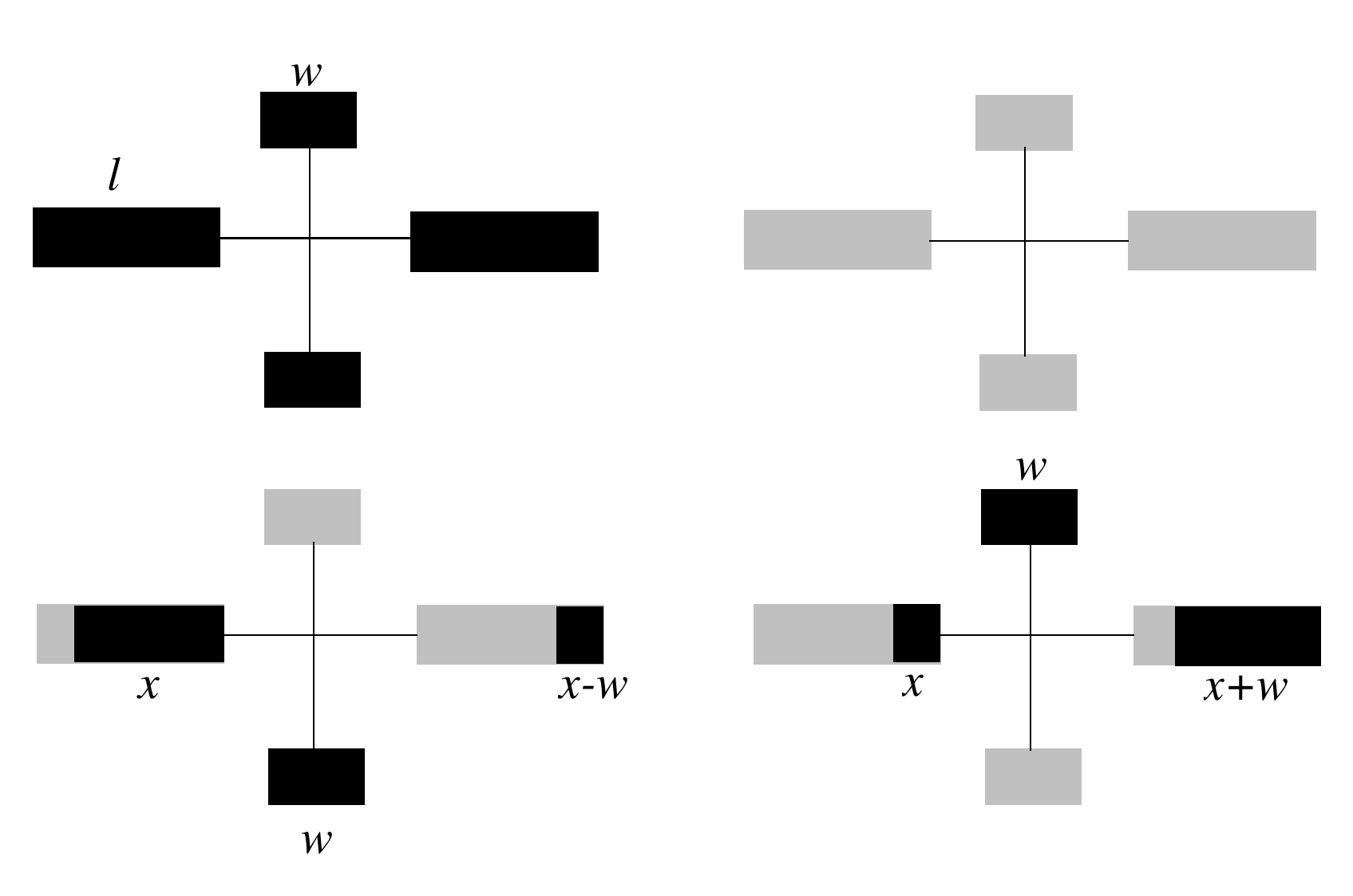}}
\end{center}
\caption{Local time evolution rules for trop p-Toda after the fine splitting of the segments. All the vertices can be interpreted as the $\widehat{\mathfrak{sl}}_2$ case of the combinatorial $R: \mathbb{B}_l \times \mathbb{B}_w \rightarrow \mathbb{B}_w \times \mathbb{B}_l$.}
\label{f4}
\end{figure}

Theorem \ref{th:may2_1} assures that
given any two-colored strip $p=S_1$ there is a carrier $X = (X_1,X_2)$ of capacity $l$
which satisfies the relation $X = Y$.
Denote any such carrier by $v_l (p)$.
We note that the above proof contains a practical procedure to
construct a $v_l (p)$:
when $\wt (p) \geq 0$ (resp.~$\wt (p) < 0$), set $(X_1, X_2) = (l,0)$ (resp.~$(X_1, X_2) = (0,l)$) 
in Fig.~\ref{fig:xp2} and one gets $v_l (p) = (Y_1, Y_2)$.
\section{Time Evolutions and Their Commutativity}\label{sec:7}
With $v_l (p)$ the two-colored strip $T_l (p) = S_2$ is uniquely determined.
Denote this relation by $v_l (p) \times p = T_l (p) \times v_l (p)$.
Then, the strip $T_k T_l (p)$ for any $k \in \R_{>0}$ is determined by the relation
$v_k (T_l (p)) \times T_l (p) = T_k  T_l (p) \times v_k (T_l (p))$.
In this section we prove the commutativity of the time evolutions $[T_l, T_k] = 0$.
\begin{Theorem}\label{t:july15f}
The commutativity $T_lT_k(p) = T_kT_l(p)$ holds.
\end{Theorem}
\proof
We adopt the splitting procedure in the proof of Theorem \ref{th:may2_1}.
Not only the top two types of the vertices in Fig.~\ref{f2}, but also the bottom two types can be split trivially.
Hence one can consider a {\em common fine splitting} of more than two strips which are related by 
the (forward and backward) time evolutions.

Take a common fine splitting of $p, T_l (p), T_k T_l (p)$, and $T_k(p)$,
after which no two-colored segments exist and all the vertices involved are of the types in Fig.~\ref{f4}.
In what follows, we call the sequences of the fine split segments by the names of their original strips $p, T_l (p)$, etc.
Given $p$, let $v_l (p), {v}_{k}(T_l(p))$ be the associated carriers defined as above.
Let $\overline{v_{l}(p)}$ and $\overline{v_{k}(T_l(p))}$ be carriers defined by
$R ({v}_{k}(T_l(p)) \times {v}_{l}(p)) = \overline{v_{l}(p)} \times \overline{v_{k}(T_l(p))}$,
where $R$ is the $\widehat{\mathfrak{sl}}_2$ case of the combinatorial $R: \mathbb{B}_k \times \mathbb{B}_l \rightarrow \mathbb{B}_l \times \mathbb{B}_k$ in \S \ref{sec:3}.

See Fig.~\ref{fig:comm}
where the Yang-Baxter relation is used
to move the symbol ``$\times$'' for the $R$ from the left to the right.
\begin{figure}
\begin{center}
\setlength{\unitlength}{1.5mm}
\begin{picture}(80,55)(-20,0)

\put(0,10){\line(1,0){8}}
\put(0,20){\line(1,0){8}}

\put(11,9){$\cdots$} 
\put(12,19){$\cdots$}

\put(18,10){\line(1,0){8}}
\put(18,20){\line(1,0){8}}

\put(4,7){\line(0,1){5}}
\put(22,7){\line(0,1){5}}

\put(4,17){\line(0,1){5}}
\put(22,17){\line(0,1){5}}

\put(-13,19){$\overline{v_{k}(T_l(p))}$}
\put(-10,9){$\overline{v_{l}(p)}$}

\put(0,24){\framebox[110pt][c]{$p$}}
\put(0,14){\framebox[110pt][c]{$z$}}
\put(0,4){\framebox[110pt][c]{$w$}}

\put(28,9){$y$}
\put(28,19){$x$}

\put(32,10){\line(1,1){10}}
\put(32,20){\line(1,-1){10}}

\put(45,9){${v}_{k}(T_l(p))$}
\put(45,19){${v}_{l}(p)$}

\put(-15,14){$=$}

\put(-1.8,35){\line(1,1){10}}
\put(-1.8,45){\line(1,-1){10}}

\put(-9,35){$\overline{v_{l}(p)}$}
\put(-13,45){$\overline{v_{k}(T_l(p))}$}

\put(9,35){${v}_{k}(T_l(p))$}
\put(10,45){${v}_{l}(p)$}

\put(20,35){\line(1,0){8}}
\put(20,45){\line(1,0){8}}

\put(31,34){$\cdots$} 
\put(32,44){$\cdots$}

\put(38,35){\line(1,0){8}}
\put(38,45){\line(1,0){8}}

\put(24,32.5){\line(0,1){4.7}}
\put(42,32.5){\line(0,1){4.7}}

\put(24,42.8){\line(0,1){5}}
\put(42,42.8){\line(0,1){5}}

\put(20,50){\framebox[110pt][c]{$p$}}
\put(20,39){\framebox[110pt][c]{$T_{l}(p)$}}
\put(20,29.3){\framebox[110pt][c]{$T_{k}T_{l}(p)$}}

\put(47,35){${v}_{k}(T_l(p))$}
\put(47,45){${v}_{l}(p)$}


\end{picture}
\end{center}
\caption{Diagram to prove the commutativity $[T_l, T_k] = 0$.}
\label{fig:comm}
\end{figure}
Hence $w = T_{k}T_{l}(p)$.
Consider what $x,y,z$ should be.
Since $R ({v}_{k}(T_l(p)) \times {v}_{l}(p)) = \overline{v_{l}(p)} \times \overline{v_{k}(T_l(p))}$,
we have $x = \overline{v_{k}(T_l(p))}$ and $y = \overline{v_{l}(p)}$.

Note that the vertices between $p,z$ and $w$ in Fig.~\ref{fig:comm}
are not for the map $\overline{R}$ in \S \ref{sec:6} but for the map $R$ in \S \ref{sec:3}.
Suppose there are such vertices between $p$ and $z$ that the top segments in $p$ are white and the bottom segments in $z$
are two-colored with black right parts.
Replace them by those of the top left type in Fig.~\ref{f2}, i.e. move the black right parts to
the left ends of the bottom segments.  
Denote by $z'$ the sequence of segments obtained from $z$ after this procedure.
Then, all the vertices between $p$ and $z'$ are for the map $\overline{R}$, and we have $z' = T_k (p)$
since Theorem \ref{th:may2_1} tells $\overline{v_{k}(T_l(p))} (= x)$ is identified with ${v}_{k}(p)$.
However, since we had taken a common fine splitting of $T_l(p)$ and $T_k(p)$, no segments in $z'$ can actually be two-colored.
Hence $z = z' = T_k (p)$.

Since we had taken a common fine splitting of $T_k(p)$ and $T_kT_l(p)$,
all the vertices between 
$z$ and $w$ are for the map $\overline{R}$ as well as $R$.
Hence $\overline{v_{l}(p) } (= y)$ is identified with ${v}_{l}(T_k(p))$ by Theorem \ref{th:may2_1}.
Thus we have
$T_l T_k(p) = w = T_k T_l(p)$.
\qed

\section{Discussions}\label{sec:8}
In this paper we gave a formulation of the tropical (ultradiscrete) periodic Toda lattice
with a representation by two-colored strips, and constructed a family of commuting time evolutions.
Now we discuss some related topics and future problems.

Although we have restricted ourselves into the periodic boundary case,
this formulation is also available for the non-periodic boundary case.
We note that for the non-periodic box-ball system\cite{TS} a family of commuting time evolutions was constructed in Ref.~\citen{FOY00}.
As its generalization, it is now an easy exercise to treat
the non-periodic tropical (ultradiscrete) Toda lattice in our formulation.

In Ref.~\citen{T10} the author elucidated the level set structure of pBBS.
The level set was generally decomposed into several connected components,
which are the orbits of the actions of
an abelian group generated by commuting time evolutions in pBBS.
Every connected component was identified with (the set of integer points on) a {\em torus} whose shape was determined by
a combinatorics of Bethe ansatz\cite{KTT}.
Now we have obtained a family of commuting time evolutions in
trop p-Toda,
its level set (which is now an algebraic variety) can be elucidated in the same manner as in Ref.~\citen{T10}.
The author believes that this allows us to give a clearer interpretation of
previous works on the level set of this dynamical system\cite{IT07,IT09,IT09b}.

Moreover, a kind of partition function (in the sense of statistical mechanics)
has been defined on the level set of pBBS, and
a simple explicit expression for the partition function as a $q$-series was obtained in Ref.~\citen{KT11}.
Whether a similar function on the level set of trop p-Toda 
has a non-trivial simple expression is an open problem. 

Probably, the most challenging problem is to construct a family of commuting time evolutions
in dp-Toda \eqref{i:eq:Toda-q}.
The author considers that the method of constructing commuting time evolutions in trop p-Toda \eqref{i:eq:UDpToda}
in the present paper will shed a light on this problem.


\begin{acknowledgment}
This work was supported by 
Grant-in-Aid for Scientific Research (C) 22540241 from the Japan Society for the Promotion of Science.


\end{acknowledgment}



\end{document}